\documentclass[journaonecolumnl,onecolumn,a4paper,12pt]{IEEEtran}
\ifCLASSINFOpdf
\else

\fi
\usepackage{amsmath,amsthm,amssymb}
\usepackage{amsmath,arydshln} 

\RequirePackage[T1]{fontenc}
\RequirePackage[utf8]{inputenc}

\renewenvironment{proof}{\noindent\textit{Proof}}{\hfill{$\Box$}}

\newcommand{\binomial}[2]{\left(\begin{array}{c}#1\\#2\end{array}\right)}

\newcommand{\Tr}{{\rm Tr}}

\newcommand{\ifc}{{\rm if \ }}

\newcommand{\boxtensor}{{\Box\kern-9.03pt\raise1.42pt\hbox{$\times$}}}

\newcommand{\sF}{{\mathcal F}}

\newcommand{\F}{{\mathbb{F}}}

\renewcommand{\P}{{\mathbb P}}

\newcommand{\Z}{{\mathbb Z}}

\newcommand{\be}{\begin{eqnarray}}
\newcommand{\ee}{\end{eqnarray}}
\newcommand{\nn}{{\nonumber}}
\newcommand{\dd}{\displaystyle}

\newcommand{\TODO}[1]
{\par\fbox{\UIKeyInputDownArrowbegin{minipage}{0.9\linewidth}\textbf{TODO:} #1\end{minipage}}\par}

\theoremstyle{plain}
\newtheorem{theorem}{Theorem}

\newtheorem{corollary}{Corollary}

\newtheorem{example}{Example}

\theoremstyle{definition}

\newtheorem{definition}{Definition}

\begin{document}
%
\title{Families of sequences with good family complexity and cross-correlation measure}
%
%
%

\author{Kenan~Doğan, \and
Murat~Şahin, \and
Oğuz~Yayla
\thanks{K.~ Doğan, PhD student, Mathematics Department, Ankara University, 06560, Yenimahalle, Ankara, Türkiye e-mail: knndogan@gmail.com.}%
\thanks{M.~ Şahin is with Graduate School of Natural and Applied Sciences, Ankara University, 06110, Altındağ, Ankara, Türkiye e-mail: msahin@ankara.edu.tr.}%
\thanks{O.~Yayla is with Institute of Applied
of Mathmatics, Middle East Technical University, 06800, Çankaya, Ankara, Türkiye e-mail: oguz@metu.edu.tr.}
}

\maketitle

\begin{abstract}
\noindent In this paper we study pseudorandomness of a family of sequences in terms of two measures, the family complexity ($f$-complexity) and the cross-correlation measure of order $\ell$. We consider sequences not only on binary alphabet but also on $k$-symbols ($k$-ary) alphabet. 
We first generalize some known methods on construction of the family of binary pseudorandom sequences. We prove a bound on the $f$-complexity of a large family of binary sequences of Legendre-symbols of certain irreducible polynomials.
We show that this family as well as its dual family have
both a large family complexity and a small cross-correlation measure up to a rather large order.  Next, we present another family of binary sequences 
having high $f$-complexity and low cross-correlation measure.
Then we extend the results to the family of sequences on $k$-symbols alphabet.\end{abstract}


%
\IEEEpeerreviewmaketitle

\section{Introduction}

Pseudorandom sequence is a sequence of numbers generated deterministically and looks random. It is called binary ($k$-ary or $k$-symbol) sequence if its elements are in \{-1,+1\} (resp. $\{a_1,a_2, \ldots , a_k\}$ for some numbers $a_i$). 
Pseudorandom sequences have a lot of application areas such as telecommunication, cryptography, simulation see for example \cite{DP2010,GG2005,NA2015,TW2007}. 
According to its application area, the quality of a pseudorandom sequence is evaluated in many directions. There are statistical test packages 
(for example L'Ecuyer's TESTU01 \cite{l2007testu01}, Marsaglia's Diehard \cite{marsaglia1996diehard} or the NIST battery \cite{nist}) for evaluating the quality of a pseudorandom sequence. In addition to that, there are proven theoretical results on some randomness measures that a pseudorandom sequence needs to satisfy. For instance, linear complexity, (auto)correlation, discrepancy, well-distribution and others, see \cite{Gya2013,TW2007} and references therein.  
In some cases one needs many pseudorandom binary sequences, for instance in cryptographic applications. Therefore, their randomness has to be proven by several figures of merit, e.g.~family complexity, cross-correlation, collision, distance minimum, avalanche effect and others, see \cite{S2017} and references therein.  The typical values of some randomness measures of a truly random sequence are proven in \cite{alon2007measures,cassaigne2002finite,merai2016typical}. Sequences satisfying typical values are so called \textit{good} sequences. 

After Mauduit and S\'ark\"ozy \cite{MS1997} introduced a construction method of good binary sequences by using Legendre symbol, other construction methods have been given in the literature \cite{chen2008elliptic,chen2010structure,liu2007new}.
In
2004, Goubin, Mauduit and Sárközy \cite{GMS2004}  first constructed large families of
pseudorandom binary sequences. Later, many new constructions \cite{follath2008construction,gya2004,liu2007family,mauduit2004construction,mauduit2005construction,merai2012construction} and their complexity bounds \cite{merai2016cross,merai2015improving,sarkozy2009measures} were developed, see \cite{S2017} for others.

These pseudorandom measures defined for binary sequences have been extended  to sequences of $k$-symbols \cite{gergely2003finite,mauduit2002finite,toth2012extension}
and some constructions of good $k$-symbols sequences were given in \cite{ahlswede2006large,du2014pseudorandom,gomez2012multiplicative,liu2017large,mak2014more}.

Huaning Liu and Xi Liu \cite{LiuXi2024} constructed new family of binary sequences has both small cross-correlation measure and large family complexity.
 
In this paper we study two such measures the family complexity (in short $f$-complexity) and the cross-correlation measure of order $\ell$ of families of binary and $k$-ary sequences. 

Ahlswede et al.~\cite{AKMS2003} introduced the $f$-complexity as follows.

\begin{definition}
The \textit{$f$-complexity} $C(\sF)$ of a family $\sF$ of binary sequences $E_N \in \{-1,+1\}^N$ of length $N$ 
is the greatest integer $j \geq 0$ such that for any $1 \leq i_1 < i_2< \cdots < i_j \leq N$ and any $\epsilon_1,\epsilon_2, \ldots,  \epsilon_j \in \{-1,+1\}$ 
there is a sequence $E_N = (e_1,e_2,\ldots , e_N)\in \sF$ with $$e_{i_1}=\epsilon_1,e_{i_2}=\epsilon_2, \ldots ,e_{i_j}=\epsilon_j.$$
\end{definition}

We have the trivial upper bound 
\be \label{eqn.bound_f}
2^{C(\sF)} \leq |\sF|,
\ee
where $|\sF|=F$ denotes the size of the family $\sF$.

Gyarmati et al.~\cite{GMS2014} introduced the cross-correlation measure of order $\ell$.


\begin{definition} \label{def.ccm}
The \textit{cross-correlation measure of order $\ell$} of a family $\sF$ of binary sequences $E_{i,N} = (e_{i,1},e_{i,2},\ldots , e_{i,N}) \in \{-1+1\}^N$,  $i=1,2, \ldots , F$, 
is defined as 
$$
\Phi_\ell(\sF) = \max_{M,D,I}\left| \sum_{n=1}^{M}{e_{i_1,n+d_1} \cdots e_{i_\ell,n+d_\ell}}\right|,
$$
where $D$ denotes an $\ell$ tuple $(d_1,d_2,\ldots , d_\ell)$ of integers such that $0 \leq d_1 \leq d_2 \leq \cdots \leq d_\ell < M+d_\ell \leq N$ and $d_i \neq d_j$ 
if $E_{i,N} = E_{j,N}$ for $i \neq j$, and $I$ denotes an $\ell$ tuple $(i_1,i_2, \ldots , i_\ell)\in\{1,2,\ldots ,F\}^\ell$.
\end{definition}

For a family $\sF$ of binary sequences of length $N$ with $|\sF| < 2^{N/12}$, the expected value of the cross-correlation measure of order $\ell \leq N/(6\log_2|\sF|)$  is $\Phi_\ell(\sF) \approx \left(N\log \binomial{N}{\ell} + \ell \log{|\sF|}\right)^{1/2}$, see \cite{merai2016typical}.

We use the notation $\Phi^\circ_\ell$ for the cross correlation $\Phi_\ell$ evaluated for fixed $M = F$ and $d_i = 0$ for all $i \in \{1,2,\ldots,l\}$.

Winterhof and the author in \cite{WY2014} proved  the estimation of the $f$-complexity  $C(\sF)$ of a family $\sF$
of binary sequences 
$$E_{i,N} = (e_{i,1},e_{i,2},\ldots , e_{i,N}) \in \{-1+1\}^N, \quad i=1,\ldots,F,$$
in terms of the cross-correlation measure $\Phi_{\ell}(\overline{\sF}), \ell \in \{1,2,\ldots,\log_2{F}\}$ of the {\em dual family} $\overline{\sF}$ of binary sequences 
$$E_{i,F} = (e_{1,i},e_{2,i},\ldots , e_{F,i}) \in \{-1+1\}^{F}, \quad i=1,\ldots,N$$ 
as follows
\be \label{eq:f_ccm}
C(\sF) \geq \left\lceil \log_2{F} - \log_2{\max_{1 \leq i \leq \log_2{F}}{\Phi_{i}(\overline{\sF})}} \right\rceil -1.
\ee

In Section~\ref{sec.example} we generalize the construction of the family of binary pseudorandom sequences presented in \cite{WY2014}. We prove a bound on the $f$-complexity of the family of binary sequences of Legendre-symbols of irreducible polynomials 
$f_i(x) = x^d + a_2i^2x^{d-2} + a_3i^{3}x^{d-3} + \cdots + a_{d-2}i^{d-2}x^{2} + a_di^d$ defined as
$$\sF_1=\left\{\left(\frac{f_i(n)}{p}\right)_{n=1}^{p-1} : i=1,\ldots,p-1 \right\}.$$ 
We show that this family as well as its dual family have
both a large family complexity and a small cross-correlation measure up to a rather large order. 

In Section \ref{sec.new_family} we study another family of binary sequences 
$${\sF_2}=\left\{ \left(\frac{f(n)}{p}\right)_{n=1}^{p-1} : f \mbox{ is irreducible of degree } d \mbox{ over } \F_p\right\}$$
for a positive integer $d$.
We prove that $\sF_2$ has high $f$-complexity and low cross-correlation measure by using \eqref{eq:f_ccm}.
We note that the roots of an irreducible polynomial are called to be \textit{conjugate} to each other.

In Section \ref{sec.ccm_k} we extend the relation \eqref{eq:f_ccm} to the family of sequences on $k$-symbols alphabet. Finally, we show in Section \ref{sec.new_family_k} that extension of the family $\sF_2$ to $k$-symbols alphabet has also  high $f$-complexity and low cross-correlation measure.

Throughout the paper, the notations $U \ll V$ and $U = O(V)$ are equivalent to the statement that $\vert U \vert \leq cV$ holds with some positive constant $c$. Moreover, the notation $f(n) = o(1)$ is equivalent to $\lim_{n \to \infty}{f(n)} = 0.$ 

\iffalse
\section{A relation between family complexity and cross-correlation measure}\label{sec.ccm}
The authors \cite{WY2014} recently proved the following relationship between the $f$-complexity of a family of binary sequences and the cross-correlation measure of its dual family. For completeness we give the sketch of the proof.

\begin{theorem}
\label{thm.f_ccm}
Let $\sF$ be a family of binary sequences $(e_{k,1},e_{k,2},\ldots ,e_{k,N}) \in \{-1,+1\}^N$ for $k=1,2,\ldots ,F$ and $\overline{\sF}$ its \textit{dual family} of binary sequences 
$(e_{1,n},e_{2,n},\ldots ,e_{F,n}) \in \{-1,+1\}^{F}$ for $n=1,2,\ldots ,N$. Then we have
$$
C(\sF) \geq \left\lceil \log_2{F} - \log_2{\max_{1 \leq i \leq \log_2{F}}{\Phi_{i}(\overline{\sF})}} \right\rceil -1,
$$
where $\log_2$ denotes the binary logarithm.
\end{theorem}
\begin{proof}
Assume that a specification 
\be \label{eqn.spec}
e_{k,n_1}=\epsilon_1,e_{k,n_2}=\epsilon_2,\ldots ,e_{k,n_j}=\epsilon_j
\ee
occurs in the family $\sF$ for some $k \in \{1,2,\ldots ,F\}$.
We will prove a bound on the integer $j$ in (\ref{eqn.spec}). 
We know that
$$\frac{1}{2}(1+\epsilon_ie_{k,n_i}) = \left\lbrace \begin{array}{ll}1 & \ifc e_{k,n_i}=\epsilon_i,\\
0 & \ifc e_{k,n_i}=-\epsilon_i.
\end{array}
\right.$$
Then the number $A$ of sequences in $\sF$ satisfying (\ref{eqn.spec}) equals 
\be \nn
\begin{array}{lll}
A &=& \dd \sum_{k=1}^{F}{\frac{1}{2^j} \prod_{i=1}^{j}{(1+\epsilon_ie_{k,n_i})}}\\[.5cm]
&=&\dd \frac{1}{2^j}\sum_{k=1}^{F}{\left[1+ \sum_{{\ell}=1}^{j}{\sum_{1 \leq  i_1 < \cdots < i_{\ell} \leq j}{\epsilon_{i_1}\cdots \epsilon_{i_{\ell}}e_{k,n_{i_1}} \cdots e_{k,n_{i_{\ell}}}}}\right]}\\[.5cm]&=&\dd \frac{1}{2^j}\left[F+ \sum_{{\ell}=1}^{j}{\sum_{1 \leq i_1 < \cdots < i_{\ell} \leq j}{\epsilon_{i_1}\cdots \epsilon_{i_{\ell}} \sum_{k=1}^{F}{e_{k,n_{i_1}} \cdots e_{k,n_{i_{\ell}}}}}} \right]\\[.5cm]
&\geq &\dd \frac{1}{2^j}\left[F- \sum_{{\ell}=1}^{j}{\sum_{1 \leq i_1 < \cdots < i_{\ell} \leq j}{\left|\sum_{k=1}^{F}{e_{k,n_{i_1}} \cdots e_{k,n_{i_{\ell}}}}\right|}} \right]\\[.5cm]
&\geq &
\dd \frac{1}{2^j}\left[F- \sum_{{\ell}=1}^{j}{\binomial{j}{{\ell}}{\Phi_{\ell}(\overline{\sF})}} \right].
\end{array}
\ee

then there exists at least one sequence in $\sF$ satisfying (\ref{eqn.spec}). 

By (\ref{eqn.bound_f}) we may assume $j \leq \log_2F$, and so we get
\be \nn
\begin{array}{lll}
\dd \sum\limits_{\ell=1}^{j}{\binomial{j}{{\ell}}{\Phi_{\ell}(\overline{\sF})}} & \leq &   
 2^j \max\limits_{1 \leq i \leq \log_2 F}\Phi_{i}(\overline{\sF}).
\end{array}
\ee
Therefore for all integers $j \geq 0$ satisfying $$ j < \log_2{F} - \log_2{\max_{1 \leq \ell \leq \log_2{F}}{\Phi_{\ell}(\overline{\sF})}}$$
we have $A > 0$ which completes the proof.
\end{proof}

We note that in some cases it is not easy to find a bound for $\Phi_{i}(\overline{\sF})$. For these cases one can replace $\Phi_{i}(\overline{\sF})$ in Theorem \ref{thm.f_ccm} with $\Phi^\circ_{i}(\overline{\sF})$ for simplicity. 
\fi

\section{A family and its dual with bounded cross-correlation and family complexity measures}\label{sec.example}
In this section we present an extension of the family of sequences given in \cite{WY2014}, where a family of sequences of Legendre symbols with irreducible quadratic polynomials and its dual family were given. It was shown that both families have high family complexity and small cross-correlation
measure up to a large order $\ell$. Namely, for $p>2$ a prime and $b$ a quadratic nonresidue modulo $p$ they study the following family $\sF$ and its dual family $\overline{\sF}$:
 $${\sF}=\left\{ \left(\frac{n^2-bi^2}{p}\right)_{i=1}^{(p-1)/2} : n=1,\ldots,(p-1)/2\right\},$$
and they show
\begin{equation}\label{Ql}
 \Phi_k({\sF})\ll k p^{1/2}\log p\quad\mbox{ and }\quad
\Phi_k(\overline{\sF})\ll k p^{1/2}\log p
\end{equation}
for each integer $k=1,2,\ldots$ 
Then \eqref{eq:f_ccm} immediately implies
$$C({\sF})\ge \left(\frac{1}{2}-o(1)\right) \frac{\log p}{\log 2} \quad \mbox{and}\quad
C(\overline{\sF})\ge \left(\frac{1}{2}-o(1)\right) \frac{\log p}{\log 2}.$$

We now present an extension of this result to higher degree polynomials over prime finite fields. Note that for these families we also get analog bounds for their duals.

Let $p > 2$ be a prime number, $d \geq 5$ and $\Omega_{p,d}$ be a set of irreducible polynomials over $\F_p$ of degree $d$ defined  as 
$$\Omega_{p,d} =\{ f(x) = x^d + a_2x^{d-2} + a_3x^{d-3} + \cdots + a_{d-2}x^{2} + a_d \in \F_p[x], a_2,a_3 \neq 0\}.$$

\begin{theorem} \label{thm.ext}
 Let $\sF_f$ be a family of binary sequences for some $f \in \Omega_{p,d}$ defined as 
$${\sF_f}=\left\{ \left(\frac{f_i(n)}{p}\right)_{n=1}^{p-1} : i=1,\ldots,p-1\right\},$$
where  $f_i(X) = i^df(X/i)$
for  $i \in \{1,2,\ldots , p-1 \}$ and $p \nmid d$. Let $\overline{\sF_f}$ be the dual of $\sF_f$.
Then we have
 $$\Phi_k({\sF_f})\ll dk p^{1/2}\log p  \mbox{ and }  \Phi_k(\overline{\sF_f})\ll dk p^{1/2}\log p$$
for each integer $k \in \{1,2,\ldots , p-1\}$ and
$$C({\sF_f})\ge \left(\frac{1}{2}-o(1)\right) \frac{\log(p/d^2)}{\log 2}  \mbox{ and }  C(\overline{\sF_f})\ge \left(\frac{1}{2}-o(1)\right) \frac{\log(p/d^2)}{\log 2}$$
\end{theorem}
\begin{proof} Since otherwise the crosscorrelation values, bounded by  $p-1$ the length of the sequences, become greater than $p$; we may assume $d < p^{1/2}/2$. We note that $f(X,Y) = f_Y(X)$  is an homogeneous polynomial of degree $d$. Thus it is enough to choose an irreducible polynomial $$f(X) = X^d + a_2X^{d-2} + a_3X^{d-3} + \cdots + a_{d-2}X^{2} + a_d \in \F_p[X]$$
such that $a_2,a_3\not\equiv 0 \mod p$. It is clear that each $f_i$ is irreducible for $i \in \{1,2,\ldots p-1\}$ whenever $f(X)$ is irreducible.

According to Definition \ref{def.ccm} we need to estimate %
$$
\dd \left \lvert\sum_{n=1}^M{\left( \frac{f_{i_1}(n+d_1)}{p} \right) \cdots \left( \frac{f_{i_k}(n+d_k)}{p}  \right)}\right \rvert = \dd \left \lvert\sum_{n=1}^M{\left( \frac{f_{i_1}(n+d_1) \cdots f_{i_k}(n+d_k)}{p}  \right)}\right \rvert.
$$
We will first show that $$h(X):= f_{i_1}(X+d_1) \cdots f_{i_k}(X+d_k) $$ is a monic square-free polynomial and then apply Weil's Theorem, see \cite{iwko,ti,wi}. Since each $f_{i_j},\ j =1,2,\ldots ,k$ is an irreducible polynomial it is enough to show that they are distinct from each other. Assume that $f_{i_j}(X+d_j) = f_{i_\ell}(X+d_\ell)$ for some $j,\ell =1,2,\ldots ,k$. Then by comparing the coefficients  of the term $X^{d-1}$ we have $d_j = d_\ell$ since $p \nmid d$. Hence  we have the equality $f_{i_j}(X) = f_{i_\ell}(X)$. But then by comparing the coefficients of the terms $X^{d-2}$ and $X^{d-3}$ we have 
$$a_{2}i_j^2 = a_{2}i_\ell^2 \mbox{ and } a_{3}i_j^3 = a_{3}i_\ell^3.$$
Since $a_{2}$ and $a_{3}$ are non zero we have
$$i_j^2 = i_\ell^2 \mbox{ and } i_j^3 = i_\ell^3.$$
This implies that $i_j = i_\ell$, a contradiction. Therefore $h$ is a square-free polynomial. Since the degree of $h(x)$ is $dk$ then the following holds 
$$\Phi_k({\sF_f})\ll dk p^{1/2}\log p.$$ Similarly we can show that  
$$\Phi_k(\overline{\sF_f})\ll dk p^{1/2}\log p$$ for $k \in \{1,2,\ldots , p-1\}$.
Next we use \eqref{eq:f_ccm} to get the bounds on the family complexity.
\be\nn
\begin{array}{lcl}
C(\sF_f) &\geq & \dd \log_2\frac{F}{\max_{1 \leq \ell \leq \log_2{F}}{\Phi^\circ_{\ell}(\overline{\sF_f})}}\\[.4cm]
&\geq& \dd \log_2\frac{p-1}{d \,k\, p^{1/2}\, \log p} 
\\[.4cm]

&=& \dd \log_2\frac{p-1}{d \,(\log_2 p)\, p^{1/2}\, \log p} \\

&\geq& \dd \log_2\frac{p^{1/2}}{d \,(\log_2 p)\, \log p} \\

&\geq& \dd \log_2\frac{p^{1/2}}{d} - \log_2{log^2 p} \\

&\geq& \dd \frac{1}{2}\log_2\frac{p}{d^2} - \log_2{log^2 p} \\

& \geq & \dd (\frac{1}{2}-o(1))\frac{\log{p}/d^2}{\log 2}

\end{array}
\ee

\end{proof}

\begin{example}
    Let $d=5$ and $p=11$. The polynomial $f=x^5 + x^3 + 2x^2 + 3$ is in the set $\Omega_{11,5}$. The irreducible polynomials generated by $f_i(X) = i^5 f(X/i)$ for  $i \in \{1,2,\ldots , 10 \}$ are\\
    $f_1(x)=x^5 + x^3 + 2x^2 + 3$, $f_2(x)=x^5 + 4x^3 + 5x^2 + 8$, $f_3(x)=x^5 + 9x^3 + 10x^2 + 3$, $f_4(x)=x^5 + 5x^3 + 7x^2 + 3$, $f_5(x)=x^5 + 3x^3 + 8x^2 + 3$, $f_6(x)=x^5 + 3x^3 + 3x^2 + 8$, $f_7(x)=x^5 + 5x^3 + 4x^2 + 8$, $f_8(x)=x^5 + 9x^3 + x^2 + 8$, $f_9(x)=x^5 + 4x^3 + 6x^2 + 3$, $f_{10}(x)=x^5 + x^3 + 9x^2 + 8$.\\
    \indent The sequences generated by the polynomials above are $[[-1, -1, 1, 1, 1, 1, 1, 1, 1, 1]$, $[-1, 1, -1, 1, -1, \\ -1, -1, -1, -1, -1]$, $[1, 1, -1, 1, 1, -1, 1, 1, 1, 1]$, $[1, 1, 1, -1, 1, 1, 1, -1, 1, 1]$, $[1, 1, 1, 1, -1, 1, 1, 1, 1, -1]$,\\ $[1, -1, -1, -1, -1, 1, -1, -1, -1, -1]$, $[-1, -1, 1, -1, -1, -1, 1, -1, -1, -1]$, $[-1, -1, -1, -1, 1, -1,\\ -1, 1, -1, -1]$, $[1, 1, 1, 1, 1, 1, -1, 1, -1, 1]$, $[-1, -1, -1, -1, -1, -1, -1, -1, 1, 1]]$ \\
    \indent This sequence family does not have the complexity 3, since the tuples $\{[1, 1, 1]$, $[1, 1, -1]$, $[1, -1, 1]$, $[1, -1, -1]$, $[-1, 1, 1]$, $[-1, 1, -1]$, $[-1, -1, 1]$, $[-1, -1, -1]\}$ in the vector space $\mathbb F_{2}^3$ are not in the all possible 3 tuples of the sequence family.\\
    \indent For example,  the first three bits of sequence family are $[[-1, -1, 1],$ $[-1, 1, -1],$ $[1, 1, -1],$ $[1, 1, 1],$ $[1, 1, 1],$ $[1, -1, -1],$ $[-1, -1, 1],$ $[-1, -1, -1],$ $[1, 1, 1],$ $[-1, -1, -1]]$. However, this multi list does not contain all the vectors in $\mathbb F_{2}^3$.\\
    \indent On the other hand, the complexity-2 is satisfied. \\
    \indent The cross-correlation measure of order 5 of the family gets the maximum value 10 with M=10, I=[2,6,7,8,10] and D=[0,0,0,0,0]. \\
    \indent The dual family $\overline{\sF_f}$ gets the cross-correlation measure 9 of order 5 with I=[3,4,6,9,10] and D=[0,0,0,1,1]. The complexity of the dual family is 1. 
    
\end{example}


We note that each two sequences in $\sF_f$ (resp.~$\overline{\sF_f}$) given in Theorem \ref{thm.ext} are distinct as  $\Phi_2({\sF_f}) < p$ (resp.~$\Phi_2(\overline{\sF_f})<p$).
Hence, we have the family size $|\sF_f| = p$ for the family. 
In the next result, we give  an upper bound on  the number $\#\{\sF_f \vert f \in \Omega_{p,n}\}$ of distinct families. This result is a direct consequence from the paper \cite{ccakirouglu2022number}. Before that we will give some notation. Let $C_{\alpha}:x(y^p+y)= \alpha(x^2+1)$ be curves over $\mathbb F_p$ for $\alpha \in \mathbb F_p^\times$. Let $\#C_\alpha(\mathbb F_{p^n})$ denote the number of  points $(x,y) \in \F_{p^n}$ on $C_\alpha$ and define $S_\alpha(\mathbb F_{p^n}):=\#C_\alpha(\mathbb F_{p^n})-(p^n+1)$. 
 Let $\mu$ denote the Möbius function and $[p \text{ divides }n]$ denote its truth value, i.e., $[p \text{ divides }n]:=1$ if $p \text{ divides }n$, and $[p \text{ divides }n]:=0$ otherwise.
 
\begin{corollary} \label{cor:seq}
  \[\#\{\sF_f \vert f \in \Omega_{p,n}\} < \frac1n\sum_{d\mid n, p\nmid d} \mu(d)\left(F_p(n/d)-[p \text{ divides }n]p^{n/pd}\right),\] where $$ F_p(n) = p^{n-2}+\frac{(p-1)^2}{p^2}+\frac1{p^2}\sum_{\alpha\in\mathbb F_p^\times}S_{\alpha}(\mathbb F_{p^n}).$$
\end{corollary}
\begin{proof}
    The family $\sF$ is constructed by using irreducible polynomials $f \in \Omega_{p,n}$. 
    By \cite[Theorem 1]{ccakirouglu2022number},  we get the  number of irreducible polynomials in terms of $F_p(n)$. Then, \cite[Theorem 5]{ccakirouglu2022number} gives the result.
\end{proof}

\section{A large family of sequences with low cross correlation and high family complexity} \label{sec.new_family}
Now we construct a larger family with both small cross-correlation measure and high $f$-complexity. However, for these families of sequences we cannot say anything about their duals. 

\begin{theorem} \label{thm.ext_2}
Let $p > 2$ be a prime number, $d \in \Z^+$ and $p \nmid d$. Let $\Omega_d$ be the set defined as 
$$ \Omega_d = \{ f(X) = X^d +  a_{2}X^{d-2}+ \cdots + a_d \in \F_p[X] : f \mbox{ is irreducible over } \F_p\}.$$
Let a family $\sF_d$ of binary sequences be defined as 
$${\sF_d}=\left\{ \left(\frac{f(n)}{p}\right)_{n=1}^{p-1} : f \in \Omega_d \right\}.$$
Then,
\be \label{eqn.f_ext} 
C({\sF_d})\ge \dd (\frac{1}{2}-o(1))\frac{\log{p^{d-2}}/d^2}{\log 2}.
\ee The family size equals  $$ |\sF_d| = \frac{p^{d-1}}{d} - O(p^{\lfloor d/2 \rfloor}) $$ for $3 \leq d < p^{1/2}/2$.
\end{theorem}
\begin{proof}
It is clear by the Weil bound that each irreducible polynomial generates a distinct sequence in $\sF$. Yucas \cite{Yucas2006} proved that number of irreducible polynomials over $\F_p$ of degree $d$ with $p \nmid d$ and trace zero equals $$\frac{1}{dp}\sum_{t \mid d} {\mu(t) p^{d/t}}.$$ 
By doing calculations
$$\sum_{i=1}^{d/2}{p^i}=p(\frac{p^{d/2}-1}{p-1})=\frac{p}{p-1}p^{d/2}-\frac{p}{p-1}$$
With smallest p=3, we see that
$$\frac{1}{dp}\sum_{t \mid d} {\mu(t) p^{d/t}} \geq \frac{p^d}{dp} - \sum_{i=1}^{\lfloor d/2 \rfloor}{p^i} \geq \frac{p^{d-1}}{d} - \frac{3}{2}p^{\lfloor d/2 \rfloor}.$$
 Hence,$$ |\sF_d| = \frac{p^{d-1}}{d} - O(p^{\lfloor d/2 \rfloor}) $$ and we have proved the size of family.

Next, we prove the bound on family complexity by using \eqref{eq:f_ccm} with $\Phi^\circ_{k}$ instead $\Phi_{k}$. Because estimating $\Phi^\circ_{k}$ is easier in this case. 
In order to calculate $\Phi^\circ_{k}(\overline{\sF})$ we need to estimate 
\be \nn
V = \dd \left \lvert\sum_{f \in \Omega_d}{\left( \frac{f(i_1)}{p} \right) \cdots \left( \frac{f(i_k)}{p}  \right)}\right \rvert, \quad 1\leq i_1 < i_2 < \cdots <i_k \leq p-1.
\ee
Note that $f(X) \in \Omega_d$ if and only if $f(X) = (X - \beta)(X - \beta^p) \cdots  (X - \beta^{p^{d-1}})$ for some $\beta \in \F_{p^d}$ with $\Tr(\beta) = 0$ and $\beta \not\in \F_{p^t}$ for any $t \mid d, t <d$. Hence we rewrite $V$ as
\be \nn
\begin{array}{lcl}
V&=&  \dd \frac{1}{d} \left \lvert\sum_{\substack{\beta \in \F_{p^d} \\  \F_{p^d} = \F_p(\beta) \\ \Tr(\beta)=0}}{\left( \frac{ (i_1 - \beta)(i_1 - \beta^p) \cdots  (i_1 - \beta^{p^{d-1}})}{p} \right) \cdots \left( \frac{ (i_k - \beta)(i_k - \beta^p) \cdots  (i_k - \beta^{p^{d-1}})}{p}  \right)}\right \rvert \\[.5cm]
&=&  \dd \frac{1}{d} \left \lvert\sum_{\substack{\beta \in \F_{p^d} \\ \F_{p^d} = \F_p(\beta)  \\ \Tr(\beta)=0}}{\left( \frac{ N(i_1 - \beta)}{p} \right) \cdots \left( \frac{ N(i_k - \beta)}{p}  \right)}\right \rvert, \\[.5cm]
\end{array}
\ee
where $N$ is the norm function from $\F_{p^d}$ to $\F_p$. We note that $\chi(\alpha) = \left(\frac{N(\alpha)}{p} \right)$ is the quadratic character of $\F_{p^d}$\ and the number of elements $\alpha \in \F_{p^t},\ t \mid d$ and  $t <d$ but $\alpha \not\in \F_{p^d}$ is at most 
$$\sum_{t \mid d, t <d} {p^{t}}  \leq \frac{3}{2} p^{\lfloor d/2 \rfloor}.$$
Thus we can estimate $V$ as follows:
\be \nn
\begin{array}{lcl}
V&\leq &  \dd \frac{1}{d} \left \lvert\sum_{\substack{\beta \in \F_{p^d} \\  \Tr(\beta)=0}}{\chi((i_1 - \beta) \cdots (i_k - \beta)) }\right \rvert + O(p^{d/2}/d)\\[.6cm]
&\leq & \dd \frac{1}{dp} \left \lvert\sum_{\alpha \in \F_{p^d}}{ \chi((i_1 - \alpha^p + \alpha) \cdots (i_k - \alpha^p + \alpha)) }\right \rvert + O(p^{d/2}/d)\\[.6cm]
&\leq & \dd \frac{pk\,p^{d/2}\log p }{dp} + O(p^{d/2}/d)\quad\quad\quad \text{ (degree of the polynomial in $\chi$ is  $pk$)}\\
& = & \dd k\,p^{d/2}/d \log p + O(p^{d/2}/d).
\end{array}
\ee
The last inequality follows by Weil's Theorem \cite{weil}.  By (1) $2^{C(\sF)} \leq |\sF|=F$ and since $k=C(\sF)$ we get $ k\leq log_2{F}$.
Therefore by using \eqref{eq:f_ccm} we have
\be \nn
\begin{array}{lcl}
C(\sF_d) &\geq & \dd \log_2\frac{F}{\max_{1 \leq \ell \leq \log_2{F}}{\Phi^\circ_{\ell}(\overline{\sF_d})}}\\[.4cm]
&\geq&\dd \log_2\frac{\frac{p^{d-1}}{d} - O(p^{\lfloor d/2 \rfloor})}{k\frac{p^{d/2}}{d}\log p+O(p^{d/2}/d)} 
\\[.4cm]
& = & \dd \log_2\frac{\frac{p^{d-1}}{d} - O(p^{\lfloor d/2 \rfloor})}{\frac{1}{d}{(\log_2{F})}\,{p^{d/2}}\log p + O(p^{d/2}/d)} \quad \quad \quad \quad \text{ by (1)}
\\[.4cm] 
& = & \dd \log_2\frac{p^{d/2}(\frac{p^{d/2-1}} {d} - c_1)}{\frac{1}{d}\log_2(\frac{p^{d-1}}{d} - O(p^{\lfloor d/2 \rfloor}))\,p^{d/2}\log p+O(p^{d/2}/d)}\\[.4cm] 

& \geq & \dd \log_2\frac{\frac{p^{d/2-1}} {d} - c_1}{\frac{1}{d}\log_2(\frac{p^{d-1}}{d} - O(p^{\lfloor d/2 \rfloor}))\,\log p+c_2}\\[.4cm] 

& \geq & \dd \log_2\frac{\frac{p^{d/2-1}} {d} - c_1}{\frac{1}{d}(\log_2p^d)\log p + c_2}\\[.4cm] 

& = & \dd \log_2\frac{\frac{p^{d/2-1}} {d} - c_1}{(\log_2p)\log p + c_2}\\[.4cm] 

& \geq & \dd  \frac{1}{2}log_2{(\frac{p^{d/2-1}} {d} - c_1)^2} - \log_2({log^2 p} + c_2)\\[.4cm] 

& \geq & \dd  \frac{1}{2}log_2{(\frac{p^{d-2}} {d^2} - 2\,c_1\,\frac{p^{d/2-1}} {d} +{c_1}^2)} - \log_2({log^2 p} + c_2)\\[.4cm] 

& \geq & \dd (\frac{1}{2}-o(1))\frac{\log{p^{d-2}}/d^2}{\log 2}
\end{array}
\ee
\end{proof}

\begin{example}
Let $p=11$ and $d=5$.
$\Omega_5=\{ f(x)=x^5 + a_2 x^3 + a_3 x^2 + a_4 x + a_5 \in \mathbb{F}_{11}[x] : f \mbox{ is irreducible over }\mathbb{F}_{11} \}$. This family consists of 2640 irreducible polynomials. The f-complexity of this family is 8.  The cross-correlation measure of order 5 of the family gets the maximum value 10 with M=10, I=[2573, 244, 2118, 1629, 740] and D=[0,0,0,0,0]. \\
\end{example}

Gyarmati et.~al.~proved that the cross-correlation measure of the family given in Theorem \ref{thm.ext_2} is small and satisfies
\be \nn \label{eqn.ccm_ext} 
\Phi_k({\sF_d})\ll kdp^{1/2}\log p
\ee
for each integer $k \in \{1,2,3,\ldots , p-1\}$, see \cite[Theorem 8.14]{GMS2014}.

\section{Sequences on $k$-symbols alphabet}
\label{sec.ccm_k}
In \cite{mauduit2002finite} the correlation measure of a sequence consisting of symbols $\{a_1,a_2, \ldots , a_k\}$ is defined. We similarly extend the definition of cross-correlation measure for a family of sequences consisting of $k$-symbols in the following.

\begin{definition} \label{def.ccm_k}
The \textit{cross-correlation measure of order $\ell$} of a family $\sF$ of sequences \linebreak $E_{i,N} = (e_{i,1},e_{i,2},\ldots , e_{i,N}) \in \{a_1,a_2, \ldots , a_k\}^N$,  $i=1,2, \ldots , F$, 
is defined as 
$$
\gamma_\ell(\sF) = \max_{W,M,D,I}\left| g(\sF,W,M,D,I) - \frac{M}{k^\ell} \right|
$$
for 
$$g(\sF,W,M,D,I) = | \{ n : 1 \leq n \leq M, (e_{i_1,n+d_1}, \ldots, e_{i_\ell,n+d_\ell}) = W\} |$$ where $W \in \{a_1,a_2, \ldots , a_k\}^\ell$,
$D$ denotes an $\ell$ tuple $(d_1,d_2,\ldots , d_\ell)$ of integers such that $0 \leq d_1 \leq d_2 \leq \cdots \leq d_\ell < M+d_\ell \leq N$ and $d_i \neq d_j$ 
if $E_{i,N} = E_{j,N}$ for $i \neq j$, and $I$ denotes an $\ell$ tuple $(i_1,i_2, \ldots , i_\ell)\in\{1,2,\ldots ,F\}^\ell$.
\end{definition}

The definition of $f$-complexity C($\sF$) for a family $\sF$ of binary sequences can be directly generalized to a family of sequences consisting of $k$-symbols. 

\begin{definition}
The \textit{$f$-complexity} $C(\sF)$ of a family $\sF$ of $k$-symbol sequences $E_N \in \{a_1,a_2, \ldots , a_k\}^N$ of length $N$ 
is the greatest integer $j \geq 0$ such that for any $1 \leq i_1 < i_2< \cdots < i_j \leq N$ and any $\epsilon_1,\epsilon_2, \ldots,  \epsilon_j \in \{a_1,a_2, \ldots , a_k\}$ 
there is a sequence $E_N = (e_1,e_2,\ldots , e_N)\in \sF$ with $$e_{i_1}=\epsilon_1,e_{i_2}=\epsilon_2, \ldots ,e_{i_j}=\epsilon_j.$$
\end{definition}

Now we prove the following extension of \eqref{eq:f_ccm}.

\begin{theorem}
\label{thm.f_ccm_k}
Let $\sF$ be a family of sequences $(e_{i,1},e_{i,2},\ldots ,e_{i,N}) \in \{a_1,a_2, \ldots , a_k\}^N$ for $i=1,2,\ldots ,F$ and $\overline{\sF}$ its dual family of binary sequences 
$(e_{1,n},e_{2,n},\ldots ,e_{F,n}) \in \{a_1,a_2, \ldots , a_k\}^{F}$ for $n=1,2,\ldots ,N$. Then we have
\be \label{bound.k_g}
C(\sF) \geq \left\lceil \log_k{F} - \log_2{\max_{1 \leq i \leq \log_k{F}}{\gamma_{i}(\overline{\sF})}} \right\rceil -1
\ee
and 
\be \label{bound.k_G}
C(\sF) \geq \left\lceil \log_k{F} - \log_2{\max_{1 \leq i \leq \log_k{F}}{\Gamma_{i}(\overline{\sF})}} \right\rceil -1,
\ee
where $\log_k$ denotes the base $k$ logarithm.
\end{theorem}
\begin{proof}
Assume that for an integer $j$ a specification 
\be \label{eqn.spec_k}
e_{k,n_1}=b_1,e_{k,n_2}=b_2,\ldots ,e_{k,n_j}=b_j
\ee
for $B = (b_1,b_2, \ldots, b_j) \in \{a_1,a_2, \ldots , a_k\}^j$ occurs in the family $\sF$ for some $k \in \{1,2,\ldots ,F\}$.
Let $A$ denotes the number of sequences in $\sF$ satisfying (\ref{eqn.spec_k}).
By the definition of cross-correlation we have
\be \nn
\begin{array}{lcl}
\gamma_j(\overline{\sF}) &=& \dd \max_{W,M,D,I}\left| g(\overline{\sF},W,M,D,I) - \frac{M}{k^j} \right|\\[.4cm]
& \geq &\dd \left| g(\overline{\sF},B,F,(0,0,\ldots,0),(n_1,\ldots,n_j)) - \frac{F}{k^j} \right|\\[.4cm]
& \geq &\dd \left| A - \frac{F}{k^j} \right|\\
\end{array}
\ee
Hence, we obtain that $$A \geq \frac{F}{k^j} + \gamma_j(\overline{\sF}).$$
If $j < \log_k{F} - \log_k{\gamma_j(\overline{\sF})}$ then  there exists at least one sequence in $\sF$ satisfying (\ref{eqn.spec_k}). 
Therefore for all integers $j \geq 0$ satisfying $$ j < \log_2{F} - \log_k{\max_{1 \leq \ell \leq \log_k{F}}{\gamma_{\ell}(\overline{\sF})}}$$
we have $A > 0$ which completes the proof of (\ref{bound.k_g}). And, the proof of (\ref{bound.k_G}) is done similarly.
\end{proof}


\section{A large family of $k$-symbols sequences with low cross correlation and high family complexity} \label{sec.new_family_k}

In this section we extend the family of binary sequences that we have presented in Section \ref{sec.new_family} to the $k$-symbols alphabet. We prove the following generalization of Theorem \ref{thm.ext_2}. Its proof is similar to proof of Theorem \ref{thm.ext_2}, see also \cite[Theorem 3]{MS1997}.

\begin{theorem} \label{thm.ext_k}
Let $d$ and $p > 2$ be distinct prime numbers and $k$ be a positive integer such that $$\gcd(k,\frac{p^d-1}{p-1}) = 1.$$ Let $f_\beta$ be an irreducible polynomial  of degree $d$ over the finite field $\F_p$ such that
$$f_\beta(x) = (x - \beta)(x - \beta^p) \cdots  (x - \beta^{p^{d-1}})$$
for an element $\beta \in \F_{p^d}$.
Let a family $\sF$ of $k$-ary sequences be defined as 
$${\sF}=\left\{ \left(\chi({f_\beta(n)})\right)_{n=1}^{p-1} : \beta \in \F_{p^d}\backslash \F_p  \mbox{ and } \Tr(\beta) = 0 \right\}$$ for some character $\chi$ of order $k$.
Then we have
\be \label{eqn.ccm_ext_k} 
\gamma_\ell({\sF})\ll \ell p^{1/2}\log p
\ee
for each integer $l \in \{2,3,\ldots , p-1\}$ and
\be \label{eqn.f_ext_k} 
C({\sF})\ge \dd (\frac{d}{2}-1)\log_2{p} - \log_2{((d-1)\log_2{p})}.
\ee The family size equals $$ F = \frac{p^d-p}{dp} .$$
\end{theorem}
\begin{proof}
Let $a$ be a $k$-th root of unity and $S(a,m)$ denote 
$$ S(a,m) = \frac{1}{k}\sum_{t=1}^{k}{\overline{a}\chi(m)^t}.$$
Then we have
$$ S(a,m) = \left\lbrace \begin{array}{ll}1 & \ifc \chi(m) = a,\\
0 & \ifc \chi(m) \neq a.
\end{array}
\right.$$
Now we estimate $g(\sF,W,M,D,I)$ as follows.
\be \nn
\begin{array}{lcl}
g(\sF,W,M,D,I) &=& \dd| \{ n : 1 \leq n \leq M, (e_{i_1,n+d_1}, \ldots, e_{i_\ell,n+d_\ell}) = W \} | \\[.1cm]
&=&\dd \sum_{n=1}^{M}{\prod_{j=1}^{\ell}{S(a_j,f_{i_j}(n+d_j))}}\\[.4cm]
&=& \dd \sum_{n=1}^{M}{\prod_{j=1}^{\ell}{\frac{1}{k} \sum_{t_j = 1}^{k}{(\overline{a_j} \chi(f_{i_j}(n+d_j)))^{t_j}}}}\\[.4cm]
&=& \dd \frac{1}{k^\ell}\sum_{t_1 = 1}^{k} \cdots \sum_{t_\ell = 1}^{k}\overline{a_1^{t_1}} \cdots \overline{a_\ell^{t_\ell}} \sum_{n=1}^{M}{\chi(f_{i_1}(n+d_1))^{t_1} \cdots \chi(f_{i_\ell}(n+d_\ell))^{t_\ell}}\\[.4cm]
&=& \dd \frac{M}{k^\ell} +  \frac{1}{k^\ell}\sum_{t_1 = 1}^{k-1} \cdots \sum_{t_\ell = 1}^{k-1}\overline{a_1^{t_1} \cdots a_\ell^{t_\ell}} \sum_{n=1}^{M}{\chi(f_{i_1}(n+d_1)^{t_1} \cdots f_{i_\ell}(n+d_\ell)^{t_\ell})}\\[.4cm]
& \leq & \dd \frac{M}{k^\ell} +  \frac{1}{k^\ell}\sum_{t_1 = 1}^{k-1} \cdots \sum_{t_\ell = 1}^{k-1} \left| \sum_{n=1}^{M}{\chi(f_{i_1}(n+d_1)^{t_1} \cdots f_{i_\ell}(n+d_\ell)^{t_\ell})} \right| \\[.4cm]
\end{array} 
\ee
Now consider the polynomial  
$$f(n) = f_{i_1}(n+d_1)^{t_1} \cdots f_{i_\ell}(n+d_\ell)^{t_\ell}.$$
We will show that it is not a $k$-th power of a polynomial over $\F_p$. 
As $n+d_j<p$ for $j=1,2,\ldots,\ell$, we can write $f$ as follows 
$$f(n) = (n+d_1-\beta_{i_1})^{t_1\frac{p^d-1}{p-1}} \cdots (n+d_\ell-\beta_{i_\ell})^{t_{\ell}\frac{p^d-1}{p-1}}$$
for some $\beta_{i_1} , \ldots, \beta_{i_\ell} \in \F_{p^d} \backslash \F_p$ and $\Tr(\beta_{i_j}) = 0$ for all $j \in \{1,2, \ldots , \ell\}$. 
Then, the linear terms $(n+d_1-\beta_{i_1}), \ldots ,(n+d_\ell-\beta_{i_\ell})$ are distinct to each other. So it is enough to show that the power of each component is not divisible by $k$. And this holds as we have $t_1, \ldots , t_\ell < k$ and $\gcd(k,\frac{p^d-1}{p-1}) = 1$. 

By applying Weil Theorem to the inner character sum we obtain 
$$\left| g(\sF,W,M,D,I) - \frac{M}{k^\ell} \right| \ll  \ell p^{1/2} \log(p),$$
and by substituting this into Definition \ref{def.ccm_k} we complete the proof of (\ref{eqn.ccm_ext_k}). 

Next we prove the bound (\ref{eqn.f_ext_k}). Before this we note that family size equals the number of irreducible polynomials of degree $d$ over $\F_p$ having zero trace since they all produce a distinct sequence in $\sF$. Note that there are $\frac{p^d-p}{p}$ distinct elements in $\F_{p^d} \backslash \F_p$ having zero trace, and $d$ of them combines into an irreducible polynomial over $\F_p$. So we have  
\be \label{eqn_fam_k}
F = \frac{p^d-p}{dp}.
\ee
Now we estimate
\be \nn
\begin{array}{lcl}
g(\overline{\sF},W,F,0,I) &=& \dd| \{ n : 1 \leq n \leq F, (\overline{e}_{i_1,n+d_1}, \ldots, \overline{e}_{i_\ell,n+d_\ell}) = W \} |\\
&=& \dd \sum_{n=1}^{F}{\prod_{j=1}^{\ell}{S(a_j,\overline{f}_{i_j}(n))}} = \sum_{n=1}^{F}{\prod_{j=1}^{\ell}{S(a_j,f_{n}(i_j))}}\\
\end{array}
\ee
Similar to the proof of (\ref{eqn.ccm_ext_k}) we have
\be \nn
\begin{array}{lcl}
\dd \sum_{n=1}^{F}{\prod_{j=1}^{\ell}{S(a_j,f_{n}(i_j))}}  & \leq & \dd \frac{F}{k^\ell} +  \frac{1}{k^\ell}\sum_{t_1 = 1}^{k-1} \cdots \sum_{t_\ell = 1}^{k-1} \left| \sum_{n=1}^{F}{\chi((i_1 - \beta_n)^{t_1\frac{p^d-1}{p-1}} \cdots (i_\ell - \beta_n)^{t_\ell\frac{p^d-1}{p-1}})} \right| \\[.4cm]
& \leq &\dd \frac{F}{k^\ell} +  \frac{1}{k^\ell}\sum_{t_1 = 1}^{k-1} \cdots \sum_{t_\ell = 1}^{k-1} \left| \sum_{\substack{\beta \in \F_{p^d} \backslash \F_p,\Tr(\beta)=0\\ nonconjugate}}^{F}{\chi((i_1 - \beta)^{t_1\frac{p^d-1}{p-1}} \cdots (i_\ell - \beta)^{t_\ell\frac{p^d-1}{p-1}})} \right| \\[.4cm]
& \leq &\dd \frac{F}{k^\ell} +  \frac{1}{k^\ell}\sum_{t_1 = 1}^{k-1} \cdots \sum_{t_\ell = 1}^{k-1} \frac{1}{dp}\left| \sum_{\beta \in \F_{p^d} \backslash \F_p}^{F}{\chi((i_1 - \beta)^{t_1\frac{p^d-1}{p-1}} \cdots (i_\ell - \beta)^{t_\ell\frac{p^d-1}{p-1}})} \right|. \\[.4cm]
\end{array}
\ee
Since $\Tr(\beta) = 0$ and $\gcd(k,\frac{p^d-1}{p-1}) = 1$, the polynomial inside the character sum is not a $k$-th power. Thus by Weil Theorem we have 
$$\left| g(\overline{\sF},W,F,0,I) - \frac{F}{k^\ell} \right| \ll \frac{1}{dp}[(\ell p-1)p^{d/2}+p],$$
and so by Definition \ref{def.ccm_k} we have 
\be \label{eqn.gamma_0}
\gamma^\circ(\overline{\sF}) \ll  \frac{1}{dp}[(\ell p-1)p^{d/2}+p].
\ee
Therefore by using (\ref{eqn_fam_k}) and (\ref{eqn.gamma_0}) we obtain that
$$C(\sF) \geq  \dd \log_2\frac{F}{\max_{1 \leq \ell \leq \log_2{F}}{\gamma^\circ_{\ell}(\overline{\sF})}} \geq \dd (\frac{d}{2}-1)\log_2{p} - \log_2{((d-1)\log_2{p})}$$
as desired.
\end{proof}

\section{Conclusion}
Pseudorandom sequences are used in many practical areas and their quality is decided by statistical test packages  as well as by proved results on certain measures. 
In addition to that,  a large family of good pseudorandom sequences in terms of several directions is required in some applications. 
In this paper we studied two such measures the $f$-complexity, and the cross-correlation measure of order $\ell$ for family of sequences on binary and $k$-symbols alphabet. 
We considered two families of binary sequences of Legendre-symbols  
$$\sF_1=\left\{\left(\frac{f_i(n)}{p}\right)_{n=1}^{p-1} : i=1,\ldots,p-1 \right\}$$ for irreducible polynomials 
$f_i(x) = x^d + a_2i^2x^{d-2} + a_3i^{3}x^{d-3} + \cdots + a_{d-2}i^{d-2}x^{2} + a_di^d$
and 
$${\sF_2}=\left\{ \left(\frac{f(n)}{p}\right)_{n=1}^{p-1} : f \mbox{ is irreducible of degree } d \mbox{ over } \F_p\right\}$$
for a positive integer $d$.
We show that the families $\sF_1$ and $\sF_2$ have
both a large family complexity and a small cross-correlation measure up to a rather large order. 
Then we proved the analog results for the family of sequences on $k$-symbols alphabet, and  constructed a good family of $k$-symbols sequences.
	
\section*{Acknowledgment}
This study was initiated in 2020 and supported by the Scientific and Technological Research	Council of Turkey (TÜBİTAK) under Project No: \mbox{116R026}. The study has been updated in 2022 and since then Oğuz Yayla has been supported by Middle East Technical University - Scientific Research Projects Coordination Unit under grant number 10795.

\ifCLASSOPTIONcaptionsoff
  \newpage
\fi



\bibliographystyle{IEEEtranS}
\bibliography{family_complexity}
%



%







\end{document}